\documentclass[11pt]{article}
\usepackage{fullpage}
\usepackage{epic}
\usepackage{eepic}
\usepackage{graphicx}
\usepackage{verbatim}
\usepackage{amsmath, amssymb, amsthm}
\usepackage{algorithm}
\usepackage[noend]{algpseudocode}
\makeatletter
\def\algbackskip{\hskip-\ALG@thistlm}
\makeatother

\usepackage{color}
\usepackage{wasysym}
\usepackage{enumerate}
\usepackage[bf]{caption}
\newtheorem{theorem}{Theorem}[section]
\newtheorem{lemma}[theorem]{Lemma}

\usepackage{epstopdf}
\usepackage{psfrag}
\usepackage{times}

\usepackage{times}
\usepackage{pdfpages}

\newcommand{\remove}[1]{}

\theoremstyle{definition}

\theoremstyle{plain}

\newcommand{\qqed}{\hfill $\blacksquare$}

\newcommand{\lyxmathsym}[1]{\ifmmode\begingroup\def\b@ld{bold}
  \text{\ifx\math@version\b@ld\bfseries\fi#1}\endgroup\else#1\fi}

\usepackage{fancyhdr}

\newtheorem{def-energy-spanner}[root_id]{Definition}
\newtheorem{def-distance-spanner}[root_id]{Definition}
\begin{document}

%
\title{Privacy Vulnerabilities of Dataset Anonymization Techniques}

\author{Eyal Nussbaum and Michael Segal\\
Communication Systems Engineering Department \\ Ben Gurion
University, Beer-Sheva 84105, Israel} 
\maketitle             
\begin{abstract}
Vast amounts of information of all types are collected daily about people by
governments, corporations and individuals. The information is collected, for
example, when users register to or use on-line applications, receive health
related services, use their mobile phones, utilize search engines, or perform
common daily activities. As a result, there is an enormous quantity of
privately-owned records that describe individuals' finances, interests,
activities, and demographics. These records often include sensitive data and
may violate the privacy of the users if published. The common approach to
safeguarding user information, or data in general, is to limit access to the
storage (usually a database) by using and authentication and authorization
protocol. This way, only users with legitimate permissions can access the user
data. However in many cases the publication of user data for statistical
analysis and research can be extremely beneficial for both academic and
commercial uses, such as statistical research and recommendation systems. To
maintain user privacy when such a publication occurs many databases employ
anonymization techniques, either on the query results or the data itself (Adam
and Worthmann \cite{priv-compare}). In this paper we examine variants of $2$
such techniques, ``data perturbation'' and ``query-set-size control'', and
discuss their vulnerabilities. The data perturbation method deals with changing
the values of records in the dataset while maintaining a level of accuracy over
the resulting queries. We focus on a relatively new data perturbation method
called NeNDS \cite{cf-obfuscation} and show a possible partial
knowledge privacy attack on this method. The query-set-size control allows
publication of a query result dependent on having a minimum set size, $k$, of
records satisfying the query parameters. We show some query types relying on
this method may still be used to extract hidden information, and prove others
maintain privacy even when using multiple queries.
\end{abstract}
\newpage
\section{Introduction}
In today's world many organizations and individuals constantly gather
information about people, whether directly or indirectly. This leads to enormous
databases storing private information regarding individuals' personal and
professional life. Commonly, access to these records is limited and safeguarded
using authorization and authentication protocols. Only authorized users may
query the system for data. There are however instances in today's
global network of organizational connections, the growing demand to disseminate
and share this information is motivated by various academic, commercial and
other benefits. This information is becoming a very important resource for many
systems and corporations that may analyze the data in order to enhance and
improve their services and performance. The problem of privacy-preserving data
analysis has a long history spanning multiple disciplines. As electronic data
about individuals becomes increasingly detailed, and as technology enables ever
more powerful collection and curation of these data, the need increases for a
robust, meaningful, and mathematically rigorous definition of privacy, together
with a computationally rich class of algorithms that satisfy this definition. A
comparative analysis and discussion of such algorithms with regards to
statistical databases can be found in \cite{priv-compare}. One common practice
for publishing such data without violating privacy is applying regulations,
policies and guiding principles for the use of the data. Such regulations
usually entail data distortion for the sake of anonymizing the data. In recent
years, there has been a growing use of anonymization algorithms based on
\emph{differential privacy} introduced by Dwork et al. \cite{dif-priv}.
Differential privacy is a mathematical definition of how to measure the privacy
risk of an individual when they participate in a database. To construct a data
collection or data querying algorithm which constitutes differential privacy,
one must add some level of noise to the collected or returned data respectively.
While ensuring some level of privacy, these methods still have several issues
with regards to implementation and data usability. Sarwate and Chaudhuri
\cite{dif-challenges} discuss the challenges of differential privacy with regards to
continuous data, as well as the trade-off between privacy and utility. In some cases, the data may become
unusable after distortion. Lee and Clifton \cite{eps-choose} discuss the
difficulty of correctly implementing differential privacy with regards to the
choice of $\epsilon$ as the differential privacy factor. Due to these issues and
restrictions, other privacy preserving algorithms are still in prevalent in many
databases and statistical data querying systems. In this paper, we address
vulnerabilities of several implementations of such privacy preserving
algorithms.
\\
The vulnerability
of databases, and hence the potential avenues of attack, depend among other
things on the underlying data structure (and query behavior). The information stored in  databases also
comes in many forms, such as plain text, spatial coordinates, numeric values,
and others. Each combination of structure and data format allows for its own
specific attack and requires its own unique handling of privacy protection.
Another factor when handling privacy in databases is the type of queries
allowed (which may be dictated by the previously mentioned structure and data
format). For example, datasets with timestamp values may only allow
min/max and grouping queries, while those containing sequential numeric values
may also allow queries regarding averages, sums, and other mathematical
formulas. In Section~\ref{sec:combine} we analyze the effectiveness of
different queries using the $k$-query-set-size limitation over aggregate
functions in maintaining individual user privacy in a vehicular network.
\\
Another field where privacy concerns are a growing issue is the field of
recommendation systems. Many of these systems use the collaborative filtering
technique, in which users are required to reveal their preferences in order to
benefit from the recommendations. Su et al. \cite{cf-survey} survey these
techniques in depth. Several methods aimed at hiding and anonymizing user
data have been proposed and studied in an attempt to reduce the privacy
issues of collaborative filtering. These methods include data obfuscation,
random perturbation, data suppression and others
\cite{cf-obfuscation,cf-perturbation,cf-suppression,cf-factor}.
Most of these methods rely on experimental results alone to show effectiveness, and some have already
been shown to have weaknesses that can be exploited in order to recover the
original user data \cite{cf-derive,cf-properties}. Parameswaran and Blough
\cite{cf-obfuscation} propose a new data obfuscation technique dubbed ``Nearest
Neighbor Data Substitution'' (NeNDS). In Section \ref{sec:cf} we detail a
privacy attack on NeNDS based on partial prior information, as well as address
shortcomings in the NeNDS algorithm and propose avenues of research for its
improvement. Finally, we conclude in Section 4.
\section{Combining Queries with $k$-Limited Results}\label{sec:combine} 
The underlying data structure of a database is one of the factors in
determining the querying methods used over the database. The database logic
itself may further restrict queries, in some cases allowing for querying a
specific key and in others only returning aggregate results over a set of
values. The data type stored may also be a factor when discerning which
querying methods may be used. Numeric values can allow for mathematical queries
such as sums, averages and medians. Text fields may allow for string operations
such as ``contains'', ``starts-with'', or even regular expressions. In the same
manner, these queries may also be prohibited as they may convey information
that is meant to remain private. Other limitations may be placed on queries as
well, such as the query-set-size limitation, blocking query results in cases
where a predefined number ($k$) of record look-ups have not been reached (i.e.
the number of users/items taken into consideration by the query are less than
$k$). Venkatadri et al. \cite{facebook-pii} recently demonstrated a Privacy
attack on Facebook users by utilizing an exploit in Facebook's targeted
advertising API which similarly restricted query results containing too few
users. Using a combination of multiple queries which returned aggregate results
(or no results due to a low number of users matching the query), the
researchers were able to narrow down personally identifiable information which
was regarded as private by the users. In this section we look at such cases and
attempt to determine whether an attacker can use a combination of allowed
queries in order to extract information which the prohibited queries mean to
block. This may be done using multiple queries of the same type, or a
combination of several query types. \\

\subsection{Dataset and Query Models}
We attempt to show privacy attacks on data gathered from vehicular networks. The
gathered data is stored in a centralized database which allows a set of queries
that are designed to return meaningful information without compromising that
privacy of the users. A privacy attack is defined as access to any information
gathered from the user that was not made available from standard queries on the
database.
\subsubsection{Graph Datasets Model}\mbox{} \\
A vehicular network is comprised of $n$ unique units distributed in the real
world and are displayed on a graph $G$ as a set of vertices $V$ such that each
vertex $v_i, i=(1,2,3,\ldots,n)$ represents one vehicle at a single (discrete)
point in time $t_j$. The timestamps are measured as incremental time steps from
the system's initial measurement designated $t_0 = 0$. We consider three different
graph models:
\begin{itemize}
  \item A linear graph $G_L$ with vehicles distributed along discrete
  coordinates on the $X$ axis between $-\infty, \infty$. 
  \item A two-dimensional planar graph $G_P$ with vehicles distributed along
  discrete coordinates on the $X$ and $Y$ axis between $-\infty, \infty$.
  \item A three-dimensional cubic graph $G_C$ with vehicles distributed along
  discrete coordinates on the $X$, $Y$ and $Z$ axis between $-\infty, \infty$.
\end{itemize}

For each vehicle $v_i$ at each timestamp, the speed is measured. We denote this
$S_{i,t}$ with $t$ being a discrete value timestamp.  

\subsubsection{Query Model}\mbox{} \\
Following are the set of queries allowed over the database.
\begin{itemize}
  \item $F_{avg} = f(R,\tau) \to S_{avg}$: given a range $R$ a timestamp $\tau$,
  return the average speed $S_{avg}$ over all vehicles in the given range at
  the given time.
  \item $F_{max} = f(R,\tau) \to S_{max}$ given a range $R$ a timestamp $\tau$,
  return the max speed $S_{avg}$ over all vehicles in the given range at the
  given time.
  \item $F_{min} = f(R,\tau) \to S_{min}$ given a range $R$ a timestamp $\tau$,
  return the min speed $S_{min}$ over all vehicles in the given range at the
  given time.
  \item $F_{med} = f(R,\tau) \to S_{med}$ given a range $R$ a timestamp $\tau$,
  return the median speed $S_{med}$ over all vehicles in the given range at the
  given time.
\end{itemize}
The range $R$ is defined by a set of boundaries over the relevant graph:
\begin{itemize}
  \item $R^L$ in $G_L$: A starting coordinate $x_{start}$ and end coordinate
  $x_{end}$.
  \item $R^P$ in $G_P$: A rectangle with corners $(x_1, y_1),(x_2, y_2)$.
  \item $R^C$ in $G_C$: A box with corners $(x_1, y_1, z_1),(x_2, y_2, z_2)$.
\end{itemize}

In order to protect user privacy, all queries deal with measurements over
aggregated data so as not to indicate a single user's information. As such, the
queries only return a result if at least $k$ unique values have been recorded
for the scope over which the query has been run, where $k < n$. The value $k$ is
known to the attacker, however the number of records which were a part of each
query result is not (i.e. the attacker only knows that if a result returned
there are at least $k$ records in the requested scope $R,\tau$, but not the
exact number).

\subsection{Analysis of $F_{avg}$}

In this section we present privacy attack problems over different graphs and
queries.
\subsubsection{Linear vehicular placement}\label{sec:avg_attack}\mbox{} \\
Model: A linear graph $G_L$ with $n$ vehicles.\\
Queries: $F_{avg}$.\\
Attack: find the speed of a single vehicle $v_{target}$ at a given time
$\tau$.\\

It is easy to see that a single query will not constitute an attack. The attack
can be performed using the following algorithm:
\begin{itemize}
  \item Select a range $R^L_1$ with $x_{start_1} = x_{v_{target}}, x_{end_1} =
  \infty$.
  \item Run query $F(R^L_1,\tau)$ and denote the result $S_{avg1}$.
  \item Select a new range $R_2$ with $x_{start_2} = x_{v_{target}} + 1,
  x_{end_2} = \infty$.
  \item Run query $F(R^L_2,\tau)$ and denote the result $S_{avg2}$.
  \item Continue querying over ranges, each time incrementing $x_{start_i}, 1 <
  i < n$ until a result isn't returned. Mark the last coordinate which returned a
  result as $x_{end}$ and the result returned as $S_{avg(k)}$. Note that
  there were $k$ records in this scope.
  \item You can now backtrack over all results and calculate the speed of each
  vehicle between $x_{s_1}$ and $x_{end}$.
\end{itemize}
Denote this algorithm $ALG_{S_L}$. We can see that the runtime for this
algorithm is the number of query iterations required to find a section with
$k-1$ vehicles.

\subsubsection{2D vehicular placement}\mbox{} \\
Model: A two-dimensional planar graph $G_P$ with $n$ vehicles distributed along
discrete coordinates on the $X$ and $Y$ axis between $-\infty, \infty$.\\
Queries: $F_{avg}$.\\
Attack(1): find the speed of a single vehicle $v_{target}$ at a given time
$\tau$. \\
Attack(2): find the average speed of a set $U$ of vehicles, with the size of
$U$ smaller than $k$, $|U| < k$ at a given time $\tau$. \\
Assumptions: The values of $n$ and $k$ are known, where $n >> k$.
We first select some value $y$ on the $Y$ axis, denote this value $Y_{mid}$, 
and split $G_P$ into 3 ranges (the section above $Y_{mid}$, the
section below $Y_{mid}$, and the section containing only $Y_{mid}$):
\begin{itemize}
  \item $R^P_{top}$: $(-\infty, Y_{mid}),(\infty, \infty)$.
  \item $R^P_{bot}$: $(-\infty, Y_{mid}),(\infty, -\infty)$.
  \item $R^P_{mid}$: $(-\infty, Y_{mid}),(\infty, Y_{mid})$.
\end{itemize}
Note that both $R^P_{top}$ and $R^P_{bot}$ contain $R^P_{mid}$, and the union of
$R^P_{top}$ and $R^P_{bot}$ is the entire graph containing all vehicles
($R^P_{n} = R^P_{G_P} = R^P_{top} \cup R^P_{bot}$). We define $\hat{R}^P_{top}$
$\hat{R}^P_{bot}$ to be $R^P_{top} \cap R^P_{mid}$ and $R^P_{bot} \cap
R^P_{mid}$ respectively. See partition example in Figure \ref{fig:vehicle_net_1}.
It is important to note that due to symmetry, this partition can also be done
around some value $x$ on the $X$ axis, with the sections built around this
value $X_{mid}$.

We now perform 5 $F_{avg}$ queries on $G_P$:\\
 $S_{top} = F_{avg}(R^P_{top}, \tau)$, $S_{bot} = F_{avg}(R^P_{bot}, \tau)$,
 $\hat{S}_{top} = F_{avg}(\hat{R}^P_{top}, \tau)$, $\hat{S}_{bot} =
 F_{avg}(\hat{R}^P_{bot}, \tau)$, $S_{n} = F_{avg}(R^P_{n}, \tau)$.\\
 
\begin{figure}
\begin{center}
\includegraphics[scale=0.25]{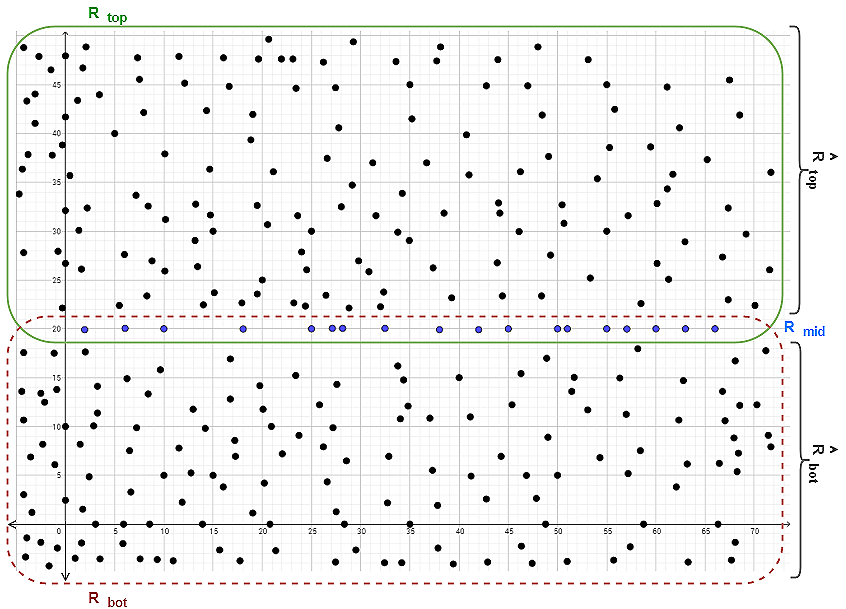}
\caption{Vehicular network $G^P$ range partition around $Y_{mid} = 20$}
\label{fig:vehicle_net_1}
\end{center}
\end{figure}
If one of the selected queries does not return a response (i.e. it contains
less than $k$ vehicles), we re-select $Y_{mid}$ and repeat the process until
all 5 queries are answered (such a value $Y_{mid}$ should exists due to the
size of $n$ and the probable distribution of vehicles).

Using the results $S_{top}, S_{bot}, \hat{S}_{top}, \hat{S}_{bot}, S_{n}$ we
wish to find  the average speed of vehicles in section $S_{mid} =
F_{avg}(R^P_{mid}, T)$, and the number of vehicles in each section: $R^P_{top},
R^P_{bot}, \hat{R}^P_{top}, \hat{R}^P_{bot}, R^P_{mid}$. The number of
vehicles in each section is a function of the section range and a given
timestamp: $C(R,\tau)$. We denote the number of vehicles in each section as
follows:
\begin{enumerate}
	\item $T = C(R^P_{top},\tau)$.
	\item $B = C(R^P_{bot},\tau)$.
	\item $M = C(R^P_{mid},\tau)$.
	\item $\hat{T} = C(\hat{R}^P_{top},\tau) = T - M$.
	\item $\hat{B} = C(\hat{R}^P_{bot},\tau) = B - M $.
\end{enumerate}

To do so, we solve the following equation system:
\begin{itemize}
  \item $n \cdot S_{n} = B \cdot S_{bot} + T \cdot S_{top} -
  M \cdot S_{mid}$.
  \item $B + T - M = n$.
  \item $T \cdot S_{top} = \hat{T} \cdot \hat{S}_{top} + M \cdot S_{mid} =
  (T-M) \cdot \hat{S}_{top} + M \cdot S_{mid}$.
  \item $B \cdot S_{bot} = \hat{B} \cdot \hat{S}_{bot} + M \cdot S_{mid} = (B-M)
  \cdot B \cdot \hat{S}_{bot} + M \cdot S_{mid}$.
\end{itemize}

Solving this system gives us the following:
\begin{itemize}
  \item $B = \frac{n \cdot (S_n - \hat{S}_{top})}{S_{bot} - \hat{S}_{top}}$.
  \item $T = \frac{n \cdot (S_n - \hat{S}_{bot})}{S_{top} - \hat{S}_{bot}}$.
  \item $M = \frac{n \cdot [S_n \cdot (\hat{S}_{top} + \hat{S}_{bot} - S_{bot} -
  S_{top}) + S_{bot} \cdot S_{top} - \hat{S}_{bot} \cdot \hat{S}_{top}]}{(\hat{S}_{bot} -
  S_{top}) \cdot S_{bot} + (S_{top} - \hat{S}_{bot}) \cdot \hat{S}_{bot}}$.
  \item $S_{mid} = \frac{(\hat{S}_{bot} \cdot \hat{S}_{top} - S_{bot} \cdot S_{top})
  \cdot S_n}{S_{bot} \cdot S_{top} + (\hat{S}_{top} +
  \hat{S}_{bot} - S_{top} - S_{bot}) \cdot S_n - \hat{S}_{bot} \cdot
  \hat{S}_{top}} + \\
  \frac{(S_{top} - \hat{S}_{top}) \cdot S_{bot} \cdot
  \hat{S}_{bot} + (S_{bot} - \hat{S}_{bot}) \cdot S_{top} \cdot \hat{S}_{top})}{S_{bot} \cdot S_{top} + (\hat{S}_{top} +
  \hat{S}_{bot} - S_{top} - S_{bot}) \cdot S_n - \hat{S}_{bot} \cdot
  \hat{S}_{top}}$.
\end{itemize}
Denote this process $ALG_{S_P}$. The runtime for this algorithm is the
equivalent to running 5 queries on the dataset, with the addition of solving
above equation system.\\

With these values we can now attempt Attack(1) and Attack(2):
If $M < k$, we have succeeded in Attack(2).
If $M > k$, we can run $ALG_{S_L}$ on $R^P_{mid}$ which represents the
boundaries of a linear graph, we can select on any vehicle with $k$
vehicles on either side of it as the target vehicle $v_{target}$ and perform
Attack(1). If $M = k$, we cannot complete either attack, so we select a new
value $Y_{mid}$ and run $ALG_{S_P}$ again. There exists an edge case of graphs
where for all values of $y$ that we can choose as $Y_{mid}$, the number of
vehicles $M$ will be equal to $k$, in which case we will be unable to perform
any attack. This scenario is, however, unlikely in the case of vehicular
networks. In addition, since we have the number of vehicles $T$ and $B$ in
$R^P_{top}$ and $R^P_{bot}$ respectively, if these values are sufficiently
large in relation to $n$, we can look at these ranges as sub-graphs of $G_P$
and run $ALG_{S_P}$ on them with $n_{T} = T$ and $n_{B} = B$.\\
It is easy to see that we can apply the same method used on the two dimensional
graph $G_P$ on the three dimensional graph $G_C$ with some minor modifications
as follows. We again select some value $y$ on the $Y$ axis, denote this value
$Y_{mid}$, and split $G_C$ into 3 ranges (the section above $Y_{mid}$, the
section below $Y_{mid}$, and the section containing only $Y_{mid}$). In this
instance, these sections are represented as cubes in the following manner:
\begin{itemize}
  \item $R^C_{top}$: $(-\infty, Y_{mid}, -\infty),(\infty, \infty, \infty)$.
  \item $R^C_{bot}$: $(-\infty, Y_{mid}, -\infty),(\infty, -\infty, \infty)$.
  \item $R^C_{mid}$: $(-\infty, Y_{mid}, -\infty),(\infty, Y_{mid}, \infty)$.
\end{itemize}
Similarly, we define $\hat{R}^C_{top}$, $\hat{R}^C_{bot}$ to be $R^C_{top} \cap
R^C_{mid}$ and $R^C_{bot} \cap R^C_{mid}$ respectively. Note that after running
our five queries on the five sections, we achieve the same linear equations as
in the two dimensional case. Solving these equations now leaves us with the
average speed over the plane defined by $R^C_{mid}$, and the number of vehicles
$M$ in this plane. As in the two dimensional case, if $M<k$ we have succeeded
in Attack(2). If $M>k$ we now have a sub-graph of $G^C$ which constitutes a two
dimensional graph $G_P$ on which we may be able to perform $ALG_{S_P}$. The
minimum size $M$ for this to be possible is $M=2k+1k$.\\
While our results, given as $ALG_{S_L}, ALG_{S_P}$ and $ALG_{S_C}$, refer to the
average speeds of vehicles in their respective graph placements, they are not
limited to speed values. The same methods can be used for any numeric value that
can be averaged over a set of vehicles in this manner, such as number of
traffic violations a vehicle has accumulated, number of accidents the
vehicle has been involved in, and so on. Any of these, when given as averages
over a set of vehicles may appear innocent and maintain high level of privacy
for an individual in the system. However, as we have shown, an individual's
data can be inferred with minimal effort by employing our methods. Of course, we
are also not limited to vehicular networks. Any data set with the same structure
of node placement in a graph will yield the same results.

\subsection{Analysis of $F_{min}, F_{max}$ and $F_{med}$}
In this section we look at possible attacks using the minimum, maximum and
median value queries over ranges in the graph as defined previously by
$F_{min}, F_{max}$ and $F_{med}$ respectively. Similar to the case of
$F_{avg}$, we define that the queries will not return a result if the target
Range $R$ at time $\tau$ contains less than $k$ individual values. In addition,
our analysis of potential attacks rests on the following set of assumptions:
\begin{itemize}
  \item The data set consists of $n$ unique values.
  \item The value $k$ is known to the attacker.
  \item In case a result is returned, the number of actual values in $(R, \tau)$
  is not known to the attacker.
  \item If $(R, \tau)$ contains an even number of values, $F_{med}$ returns the
  lower of the $2$ median values.
  \item The attacker is limited only to the $F_{min}, F_{max}$ and $F_{med}$ queries, but can perform
  any number of queries over the data set.
\end{itemize}
For simplicity, we will treat the data set as in the previous section - a
linear graph $G_L$ representing a snapshot in time $\tau$ of recorded speeds of
vehicles in a specified area. A query of type $q$ ($q$ being $min, max$ or
$med$) at time $\tau$ over a range beginning at $x_i$ and ending at $x_j$
(inclusive) will be denoted $F_q([x_i, x_j])$.

 We note that there are several special cases in which a trivial attack can be
 performed. We will address these cases before moving on to the general case.
\subsubsection{Case 1: Global Min/Max}\mbox{} \\
Since there exist a unique global minimum and global maximum in the graph, it is
easy to see that by querying over the entire graph and iteratively decreasing
the range until a new minimum/maximum is found, the vehicle with the minimum and
maximum speeds can be discovered.
\subsubsection{Case 2: $k-local$ Min/Max}\mbox{} \\
Similar to the case of a global min/max, if a vehicle has the local minimum or
maximum value with regards to his $k$ nearest neighbors then their speed can be
discovered. This is done using the same method as stated for the global min/max.
A range consisting of $k+1$ vehicles, with the outer vehicle having a min (max)
speed in that group must be found. Once found, decrease the range until a group
of size $k$ remains in its bounds. By our $k-local$ definition, the min (max)
value now changes, and the attacker knows that the previous value belongs to the
vehicle that has been removed from the range. Note that if a such a $k-local$
min/max vehicle exists in the graph, the attacker can find it given enough
queries.
\subsubsection{Case 3: $k = 3$}\mbox{} \\
In this case, since all values are defined to be unique, querying $F_{min},
F_{max}, F_{med}$ on a range containing exactly $3$ vehicles return $3$ values,
each belonging to a specific vehicle. An attacker can query over a single
coordinate at the left-most side of the graph and increase the range until a
result is returned. The first time a result is returned, the minimum $k = 3$
group size has been reached, and the attacker has the speed of each of the $3$
vehicles. Each speed cannot be attributed to a specific vehicle, but we will
denote these $3$ values $s_1, s_2, s_3$. The attacker now decreases the range's
size from the left until no result is returned, this indicates the range now
only contains $2$ vehicles. Increasing the range to the right until a result is
returned indicates that a new vehicle has been added to the range. Since all
values are unique, one of the values $s_1, s_2, s_3$ will be missing from the
results. This belongs to the left-most vehicle from the previous query results.
Continuing this method until the entire graph has been scanned will reveal the
speeds of each vehicle in the graph.
\subsubsection{The General Case: $k \geq 4$}\mbox{} \\
We show that for the general case, there exists a linear placement of vehicles
such that at least $1$ vehicle will have a speed whose value will remain hidden
from an attacker. Note that if a combination of $F_{min}, F_{max}, F_{med}$
queries can be used to attain the same results as the query $F_{avg}$, then a
privacy attack can be performed in the manner detailed in Section
\ref{sec:avg_attack}. Hozo et al.
\cite{min-max-med} devise a method to estimate the average value and variance of
a group using knowledge of only the minimum, maximum and median values. However,
for the attack described in $ALG_{S_L}$ to succeed, the actual average value is required
and not just an estimate. We use an adversarial
model and show that for any number $n$ of vehicles and any minimal query size
$k$, a vehicle arrangement can be created in which the attacker, using any
combination of the above mentioned queries, lacks the ability to discover the
speed of at least $1$ vehicle. For any value of $k \geq 4$ and $n \geq k$ we
prove this for a specific vehicle $v_1$ placed at the leftmost occupied coordinate on
the $X$ axis (denoted $x_1$). For any value of $k \geq 4$ and $n \geq 2k$ we
prove this vehicle may be at any coordinate.
\begin{lemma}\label{lemma-k-4}
Let $V$ be a set of $n$ vehicles $v_1, v_2, \ldots, v_n$ positioned along a linear
graph at coordinates $x_1, x_2, \ldots, x_n$ at time $\tau$. If $k \geq 4$, for
any value $n \geq k$ there exists a corresponding assignment of speeds $s1,
s2, \ldots, s_n$, such that the speed $s_1$ of $v_1$ cannot be determined by
any attacker with access to the $F_{min}, F_{max}$ and $F_{med}$ queries over
the graph.
\end{lemma}
\begin{proof}
We prove by induction for $k = 4$ and $n \geq k$, then extrapolate for $k \geq
4$ and $n \geq k$.
\paragraph{Show Correctness for $k=4, n=4$}\mbox{} \\
With vehicles $v_1, v_2, v_3, v_4$ positioned at $x_1, x_2, x_3, x_4$, set the
values of $s_1, s_2, s_3, s_4$ such that $s_2 > s_1 > s_3 > s_4$. Since $k = n =
4$ the queries will only return results when the range queried contains the
range $[x_1, x_n]$. It is easy to see that:
\begin{itemize}
  \item $F_{min}([x_1, x_n]) = s_4$.
  \item $F_{max}([x_1, x_n]) = s_2$.
  \item $F_{med}([x_1, x_n]) = s_3$.
\end{itemize}
As such, the value $s_1$ of $v_1$ is never revealed.
\paragraph{Assume Correctness for $k=4, n=N$}\mbox{} \\
Given a set of vehicles $v_1, v_2, v_3, \ldots, v_N$ positioned at coordinates
$x_1, x_2, x_3, \ldots, x_N$, assume there exists an assignment of
corresponding speeds $s_1, s_2, s_3, \ldots, s_N$ such that $s_1$ cannot be
determined by an attacker with access to any number of $F_{min}, F_{max},
F_{med}$ queries with a $k = 4$ limitation.

\paragraph{Prove for $k=4, n=N+1$}\mbox{} \\
We assign $s_1, s_2, s_3, \ldots, s_N$ such that for the subgraph $[x_1, x_N]$,
for $k = 4$, the value of $s_1$ is never revealed by any query $F_{min},
F_{max}, F_{med}$. We note $2$ properties regarding $s_{N+1}$ of the node
$v_{N+1}$, placed at $(x_{N+1} \mid x_{N+1} > x_N)$:
\begin{enumerate}
  \item There exists only $1$ queryable range, $[x_1, x_{N+1}]$, for which any
  query will take both $s_1$ and $s_{N+1}$ into consideration.
  \item Regardless of the value of $s_{N+1}$, the queries $F_{min}([x_1,
  x_{N+1}])$ and $F_{max}([x_1, x_{N+1}])$ cannot return $s_1$ as a
  result. (Otherwise, $s_1$ would have been a result of one of the queries over
  the subgraph $[x_1, x_N]$)
\end{enumerate}
Due to these properties, we must only ensure that the query\\
$F_{med}([x_1, x_{N+1}])$ does not return $s_1$ as it's result. Denote
 $s_{{med}_N}$ to be the result of $F_{med}([x_1, x_{N}])$. If $s_{{med}_N} >
 s_1$ then we set $(s_{N+1} \mid s_{N+1} > s_{{med}_N})$ so that $F_{med}([x_1, x_{N+1}]) \geq
s_{{med}_N} > s_1$. Conversely, if $s_{{med}_N} < s_1$ then we set $(s_{N+1}
\mid s_{N+1} < s_{{med}_N})$ so that $F_{med}([x_1, x_{N+1}]) \leq s_{{med}_N} <
s_1$. We now have an assignment $s_1, s_2, s_3, \ldots, s_N, s_{N+1}$ such that
the value $s_1$ cannot be discovered by an attacker.

\paragraph{Extrapolate for $k \geq 4, n \geq k$}\mbox{} \\
The parameter $k$ is defined as the minimum number of vehicles required to be in
a range in order for a result to be returned. For any value of $(n \mid n > k)$ 
increasing the value of $k$ only reduces the number of available queries that
will return a result. Since it holds that there exists an assignment $s_1, s_2,
s_3, \ldots, s_n$ such that $s_1$ cannot be discovered for $k = 4$, then
setting $k = 5$ for the same assignment will not give any new information to
the attacker and $s_1$ will remain unknown. It can be seen that this is true
for any value $k$ such that $n \geq k \geq 4$.
\end{proof}
While Lemma \ref{lemma-k-4} holds for any value of $k \geq 4$ and $n \geq k$,
such an assignment, where a specific node is deterministically undiscoverable,
is susceptible to prior knowledge attacks. In addition, in most real world
cases, the value of $k$ is chosen to be on a level of magnitude lower than $n$
as to allow for many queries. We show that for these cases, specifically any
case where $k \geq 4$ and $n \geq 2k$, the vehicle whose speed is never returned
by any query can be chosen as any vehicle by the adversary.
\begin{lemma}\label{lemma-2k}
Let $V$ be a set of $n$ vehicles $v_1, v_2, \ldots, v_n$ positioned along a linear
graph at coordinates $x_1, x_2, \ldots, x_n$ at time $\tau$. If $k \geq 4$, for
any value $n \geq 2k$ there exists a corresponding assignment of speeds $s1,
s2, \ldots, s_n$, such that there exists a node $v_j$ with speed $s_j$ which 
cannot be determined by any attacker with access to the $F_{min}, F_{max}$ and
$F_{med}$ queries over the graph.
\end{lemma}
\begin{proof}
We prove by induction for $k = 4$ and $n \geq 2k$, then extrapolate for $k \geq
4$ and $n \geq 2k$.
\paragraph{Show Correctness for $k=4, n=8$}\mbox{} \\

With vehicles $v_1, v_2, v_3, \ldots, v_8$ positioned at $x_1, x_2, x_3, \ldots, 
x_8$, set the values of $s_1, s_2, s_3, \ldots, s_8$ such that $s_8 > s_4 > s_3
> s_7 > s_6 > s_1 > s_2 > s_5$.  The value of $s_3$ cannot be
determined by an attacker even by running all possible query combinations on the graph. 
The results of all such possible queries can be see in Table \ref{tbl:table-2k}.

\begin{table}[ht]
  \caption{All Possible Results of $F_{min}, F_{max}$ and $F_{mid}$ with $k=4$
  and $n=8$.}
  \label{tbl:table-2k}  
  \begin{flushleft}
  \begin{tabular}{|l|c|c|c|c|c|}
    \cline{2-6} 
    \multicolumn{1}{}{}&\multicolumn{5}{ |c| }{Range Containing $4$ Vehicles}\\
    \cline{2-6} 
    \multicolumn{1}{c|}{}&$[x_1,x_4]$&$[x_2,x_5]$&$[x_3,x_6]$&$[x_4,x_7]$&$[x_5,x_8]$\\
    \hline 
    $F_{min}$&$s_2$&$s_5$&$s_5$&$s_5$&$s_5$\\
    \hline 
    $F_{max}$&$s_4$&$s_4$&$s_4$&$s_4$&$s_8$\\
    \hline 
    $F_{med}$&$s_1$&$s_2$&$s_6$&$s_6$&$s_6$\\
  \hline 
\end{tabular}
  \begin{tabular}{|l|c|c|c|c|}
    \cline{2-5} 
    \multicolumn{1}{}{}&\multicolumn{4}{ |c| }{Range Containing $5$ Vehicles}\\
    \cline{2-5} 
    \multicolumn{1}{c|}{}&$[x_1,x_5]$&$[x_2,x_6]$&$[x_3,x_7]$&$[x_4,x_8]$\\
    \hline 
    $F_{min}$&$s_5$&$s_5$&$s_5$&$s_5$\\
    \hline 
    $F_{max}$&$s_4$&$s_4$&$s_4$&$s_8$\\
    \hline 
    $F_{med}$&$s_1$&$s_6$&$s_7$&$s_7$\\
  \hline 
\end{tabular}
\begin{tabular}{|l|c|c|c|}
    \cline{2-4} 
    \multicolumn{1}{}{}&\multicolumn{3}{ |c| }{Range Containing $6$ Vehicles}\\
    \cline{2-4} 
    \multicolumn{1}{c|}{}&$[x_1,x_6]$&$[x_2,x_7]$&$[x_3,x_8]$\\
    \hline 
    $F_{min}$&$s_5$&$s_5$&$s_5$\\
    \hline 
    $F_{max}$&$s_4$&$s_4$&$s_8$\\
    \hline 
    $F_{med}$&$s_1$&$s_6$&$s_7$\\
  \hline 
\end{tabular}
\begin{tabular}{|l|c|c|}
    \cline{2-3} 
    \multicolumn{1}{}{}&\multicolumn{2}{ |c| }{Range Containing $7$ Vehicles}\\
    \cline{2-3} 
    \multicolumn{1}{c|}{}&$[x_1,x_7]$&$[x_2,x_8]$\\
    \hline 
    $F_{min}$&$s_5$&$s_5$\\
    \hline 
    $F_{max}$&$s_4$&$s_8$\\
    \hline 
    $F_{med}$&$s_6$&$s_7$\\
  \hline 
\end{tabular}
\begin{tabular}{|l|c|}
    \cline{2-2} 
    \multicolumn{1}{}{}&\multicolumn{1}{ |c| }{Range Containing $8$ Vehicles}\\
    \cline{2-2} 
    \multicolumn{1}{c|}{}&$[x_1,x_8]$\\
    \hline 
    $F_{min}$&$s_5$\\
    \hline 
    $F_{max}$&$s_8$\\
    \hline 
    $F_{med}$&$s_6$\\
  \hline 
\end{tabular}
\end{flushleft}
\end{table}
\paragraph{Assume Correctness for $k=4, n=(N \mid N\geq2k$)}\mbox{} \\
Given a set of vehicles $v_1, v_2, v_3, \ldots, v_N$ positioned at coordinates
$x_1, x_2, x_3, \ldots, x_N$, assume there exists an assignment of
corresponding speeds $s_1, s_2, s_3, \ldots, s_N$ such that there exists
some value $s_j$ belonging to some vehicle $v_j$ at position $x_j$, which
cannot be determined by an attacker with access to any number of $F_{min},
F_{max}, F_{med}$ queries under a $k = 4$ limitation.

\paragraph{Prove for $k=4, n=(N+1 \mid N\geq2k)$}\mbox{} \\
We assign $s_1, s_2, s_3, \ldots, s_N$ such that for the subgraph $[x_1, x_N]$,
for $k = 4$, there exists some value of $s_j$  which is never revealed by any
query $F_{min}, F_{max}, F_{med}$. Assume $s_j \in (s_1, s_2, s_3, \ldots,
s_{N-3})$. We note $2$ properties regarding $s_{N+1}$ of the node
$v_{N+1}$, placed at $(x_{N+1} \mid x_{N+1} > x_N)$:
\begin{enumerate}
  \item Regardless of the value of $s_{N+1}$: $\forall x_i \mid x_i \in [x_1,
  x_j] \Rightarrow F_{min}([x_i, x_{N+1}]) < s_j$. (i.e. $s_j$ cannot be the
  result of any $F_{min}$ query in the range $[x_1, x_{N+1}]$)
  \item regardless of the value of $s_{N+1}$ : $\forall x_i \mid x_i \in [x_1,
  x_j] \Rightarrow F_{max}([x_i, x_{N+1}]) > s_j$. (i.e. $s_j$ cannot be the
  result of any $F_{max}$ query in the range $[x_1, x_{N+1}]$)
\end{enumerate}
Therefore, we must only assign $s_{N+1}$ such that it does not cause $s_j$ to be
the result of any $F_{med}$ query.
Define $s_{{med}_i}$ to be the result of $F_{med}([x_i, x_N])$. Due to the
properties of $F_{med}$, if $med_1 > s_j$ then $\forall i \mid 1 \leq i \leq j
\Rightarrow med_i > s_j$. Conversely, if $med_1 < s_j$ then $\forall i \mid 1
\leq i \leq j \Rightarrow med_i < s_j$. Otherwise at least one of those queries
would have returned $s_j$ as a result, which contradicts the induction assumption.
Define $s_{med}'$ to be the closest median value to $s_j$ from the previously
stated queries. 
\[s_{med}' = F_{med}([x_y, x_N]) \mid \forall i, 1 \leq i \leq j
\Rightarrow |s_j - s_{{med}_i}| \geq |s_j - s_y|\] 
We set $s_{N+1}$ to be some
uniformly distributed random value between $s_{med}'$ and $s_j$. We now look at
$F_{med}([x_i, x_{N+1}])$ and note that for any value $i \mid 1 \leq i \leq j$,
the results of $F_{med}([x_i, x_{N}])$ and $F_{med}([x_i, x_{N+1}])$ are either
the same value or adjacent values, as the speeds in the range differ by exactly
$1$ value. Since no value $s_{{med}_i}$ is adjacent to $s_j$, then $s_j$ cannot
be the result of any value $F_{med}([x_i, x_{N+1}], 1 \leq i \leq j)$. There
exist no other queries of the type $f_{med}$ which contain both $x_j$ and
$x_{N+1}$, therefore we now have an assignment $s_1, s_2, s_3, \ldots, s_N,
s_{N+1}$ such that the value $s_j$ cannot be discovered by an attacker.\\
The above holds for the assumption $s_j \in (s_1, s_2, s_3, \ldots,
s_{N-3})$. It is easy to see that due to symmetry, the case where $s_j \in
(s_4, s_5, s_6, \ldots, s_{N})$ allows us to shift all values of $S$ one
vehicle to the right, and assign the random value between $s_{med}'$ and $s_j$
to $s_1$. This completes correctness for all positions of $v_j$.

\paragraph{Extrapolate for $k \geq 4, n \geq 2k$}\mbox{} \\
Similar to \ref{lemma-k-4}, increasing $k$ for a given value of $n$ only reduces
the amount of information available to the attacker. Therefore, if a value $s_j$
exists for an assignment in a graph with $n$ vehicles under the $k = 4$
limitation (with $n \geq 2k$), it will exist for any value of $k$ such  that  $4
\leq k \leq \frac{n}{2}$.

\end{proof}
%


\section{Collaborative Filtering}\label{sec:cf}

Collaborative filtering (CF) is a technique commonly used to build personalized
recommendations on the Web. In collaborative filtering, algorithms are used to
make automatic predictions about a user's interests by compiling preferences
from several users. In order to provide personalized information to a user, the
CF system needs to be provided with sufficient information regarding his or her
preferences, behavioral characteristics, as well as demographic information of
the individual. The accuracy of the recommendations is dependent largely on how
much of this information is known to the CF system. However, this information
can prove to be extremely dangerous if it falls in the wrong hands. Several
methods aimed at hiding and anonymizing user data have been proposed and 
studied in an attempt to reduce the privacy issues of collaborative filtering.
Among these methods is the data obfuscation technique ``Nearest Neighbor Data
Substitution'' (NeNDS) proposed by Parameswaran and Blough in
\cite{cf-obfuscation}. Using this approach, items in each column of the database
are clustered into groups by closeness of their values, and a substitution
algorithm is applied to each group. The algorithm gives each item a new location
within the group such that each item now corresponds to a new row in the
original database. The relative closeness in values of the substituted items
allows for the recommendation system to maintain a good degree of approximation
when the CF algorithm is applied to obtain recommendations, while the
substitution itself offers a level of privacy by hiding the original values
associated with each individual user. In this section, we show the possibility of
a privacy attack on the substituted database by an attacker with partial
knowledge of the original data.

\subsection{The NeNDS Algorithm}
The Nearest Neighbor Data Substitution (NeNDS) technique
is a lossless data obfuscation technique that preserves the privacy of
individual data elements by substituting them with one of their Euclidean space
neighbors. NeNDS uses a permutation-based approach in which groups of similar
items undergo permutation. The permutation approach hides the original value of
a data item by substituting it with another data item that is similar to it but
not the same. NeNDs treats each column in the database as a separate dataset.
The first step in NeNDS is the creation of similar sets of items called
neighborhoods. These items contained in each neighborhood are selected in a
manner that maintains Euclidean closeness between neighbors using some distance
measuring function suited to the data. Each data set is divided into a
pre-specified number of neighborhoods. The items in each neighborhood are then
permuted in such a way that each item is displaced from its original position,
no two items undergo swapping, and the difference between the values of the
original and the obfuscated items is minimal. The number of neighbors in each
neighborhood is denoted $NH_{size}$, with $3 \leq NH_{size} \leq N$ where $N$
is the number of items in the dataset (this is due to the fact that $NH_{size}
= 1$ does not allow any permutation and $NH_{size} = 2$ is the trivial case of
swapping between 2 items and easily reversible).\\
The substitution process is performed by determining the optimal permutation
set subject to the following conditions:
\begin{itemize}
  \item No two elements in the neighborhood undergo swapping.
  \item The elements are displaced from their original position.
  \item Substitution is not performed between duplicate elements.
\end{itemize}
The permutation mapping is done by creating a tree depicting all possible
permutation paths and selecting the path with the minimal maximum distance
between any 2 substitutions. For example, we look at the case of the
neighborhood $[75, 77, 82, 70]$. The optimal path for substitution would be
$70 \rightarrow 77 \rightarrow 82 \rightarrow 75 \rightarrow 70$ with the new
neighborhood order being $[82, 70, 77, 75]$ and the maximal difference between
any 2 substituted items being $(70 \rightarrow 77$ and $82 \rightarrow 75)$.
Once the substitutions in each neighborhood is complete, the column of the
original database is replaced with column containing the new item positions.
The detailed algorithm can be found in \cite{cf-obfuscation}.
Note that this algorithm is deterministic for any given value of $NH_{size}$,
and will yield the same permutations given any original order of the original
dataset.

\subsection{Privacy Attack on NeNDS}\label{sec:cf-attack}
In this section we will show an attack on a NeNDS permutated database by an
attacker with partial knowledge of the original database, specifically the
attacker knows the original position of at least $NH_{size} - 2$ items in each
neighborhood. The attack is performed under the following assumptions:
\begin{itemize}
  \item The attacker has complete knowledge of the NeNDS algorithm.
  \item The attacker knows the neighborhood size, $NH_{size}$ used by the
  algorithm.
  \item The attacker can measure the Euclidean distance between the items in the
  database.
  \item The attacker has access to the output permutated database (i.e. the new
  positions of all items).
\end{itemize}
We will show the attack for a single dataset (column), however since the
algorithm is performed independently for each dataset, this can be extended to
the entire database. For a given dataset of size $n$, we define the following
notations:
\begin{itemize}
  \item Let $X$ be the original dataset $[x_1, x_2, \ldots, x_n]$.
  \item Let $Y$ be the NeNDS obfuscated dataset $[y_1, y_2, \ldots, y_n]$.
  \item Let $X_p$ be the original data items in the $p^{th}$ neighborhood,
  $[x_{p1}, x_{p2}, \ldots, x_{pn}]$.
  \item Let $Y_p$ be the obfuscated data items in the $p^{th}$ neighborhood,
  $[y_{p1}, y_{p2}, \ldots, y_{pn}]$.
  \item Let $u_{p1}, u_{p2}$ be the 2 items in $X_p$ whose original position is
  unknown to the attacker.
\end{itemize}
The attack is successful if the attacker can determine the original position in
$X$ of $u_{p1}$ and $u_{p2}$ for all values of $p, 1 \leq p \leq
\frac{n}{NH_{size}}$.
\subsubsection{The Case of $NH_{size} = 3$}\mbox{}\\
We look at the simple case of the minimal neighborhood size, $NH_{size} = 3$. In
this case, we have for each value of $p$ the neighborhood $[x_{p1}, x_{p2},
x_{p3}]$. The attacker can only know the location of 1 of these items. Assume,
without loss of generality, that the attacker knows the position of $x_{p1}$,
and as such the original dataset to be $[x_{p1}, u_{p1}, u_{p2}]$ where both
$u_{p1}$ and $u_{p2}$ could be the original positions of $x_{p2}$ and $x_{p3}$.
We now look at the output neighborhood after the NeNDS algorithm. Due to the
restrictions of the NeNDS algorithm which require each item to be relocated and
do not allow swapping between 2 items, the resulting neighborhood $Y_p$ can
only be one of the following permutations:
\begin{enumerate}
  \item $[y_{p1}, y_{p2}, y_{p3}] = [x_{p2}, x_{p3}, x_{p1}]$.
  \item $[y_{p1}, y_{p2}, y_{p3}] = [x_{p3}, x_{p1}, x_{p2}]$.
\end{enumerate}
Any other permutation would entail leaving an item in its original position.
Assume permutation (1). The attacker can determine that the value $y_{p1}$ could
not have originally been in position $u_{p2}$ since this is the current position
of $x_{p1}$ and the algorithm does not allow swapping between 2 items.
Therefore, $u_{p2} = x_{p3}$ and $u_{p1} = x_{p2}$. Assume permutation (2).
The attacker can determine that the value $y_{p1}$ could not have
originally been in position $u_{p1}$ for the same reason, and reaches the same
conclusion - the original order for the neighborhood $p$ is $[x_{p1}, x_{p2},
x_{p3}]$.
\subsubsection{The General Case of any $NH_{size}$}\mbox{}\\
In this section we will show that the knowledge of $NH_{size} - 2$ original
value positions is enough for an attacker to learn the original positions of all
$NH_{size}$ values in a neighborhood. We define $L_o(x)$ and $L_n(x)$ for any
value $x \in X$ to be the original and new location (row) of that value
respectively. Taking some neighborhood $X_p$ in $X$, the attacker knows the
position $L_o(x_{pi})$ for $NH_{size} - 2$ values in $X_p$. For 2 values,
$u_{p1}, u_{p2}$, positions $L_o(u_{p1}), L_o(u_{p2})$ remain unknown. After
obfuscation, all new positions $L_n(y_{pi})$ are known to the attacker. With
this knowledge, since the values in the neighborhood $X_p$ are chosen by their
Euclidean closeness, the attacker learns the 2 values $[u_{p1}, u_{p2}]$ and
their new positions $[L_n(u_{p1}), L_n(u_{p2})]$. There remain 2 possible
original positions $L_o(u_{p1}), L_o(u_{p2})$ between which the attacker cannot
distinguish (i.e. each one of the values could have been at each one of the
possible positions originally).\\
We now examine the new values in $L_o(u_{p1}), L_o(u_{p2})$. There are 2 cases:
either 1 of the values is $u_{p1}$ or $u_{p2}$, or both values are from the
other values in $X_p$ whose original position is known to the attacker. Note
that the case $L_o(u_{p1}) = L_n(u_{p2}), L_o(u_{p2}) = L_n(u_{p1})$ cannot
exist since by definition of the algorithm, no 2 items undergo swapping. We now
show the attack for both cases, resulting in the discovery of the original
positions for $u_{p1}, u_{p2}$.
\paragraph{Case 1}\mbox{} \\
Assume, without loss of generality, that $u_{p1}$ resides in a position whose
original value is unknown, meaning was either $u_{p1}$ or $u_{p2}$. It is easy
to see that $L_n(u_{p1}) = L_o(u_{p2})$ since no item remains in the same
position after obfuscation. In addition, the remaining unknown position is
$L_o(u_p1)$. The attacker now knows the original position of both previously
unknown values.
\paragraph{Case 2}\mbox{} \\
In this case, both $L_o(u_{p1})$ and $L_o(u_{p2})$ now contain values whose
original position were known to the attacker. We arbitrarily define those
positions to be $L_1$ and $L_2$ and their original values $u_1$ and $u_2$
respectively. The attacker can know use the following method to backtrack the
obfuscation path and find the original positions of $u_{p1}$ and $u_{p2}$. We
look at the value currently in $L_1$ and denote this value $y_{p1}$. This was
the item in the obfuscation path immediately before $u_1$. $L_o(y_{p1})$ is
known to the attacker and contains the value that was in the obfuscation path
before $y_{p1}$. Denote this value $y_{p2}$. We now continue this backtracking
of the path by examining the value in $L_o(y_{p2})$ and so on until we reach on
of the values $u_{p1}, u_{p2}$. Since the path is created using a tree structure
which contains no cycles, the first unknown value we will find must correspond
to $u_2$ (as $u_1$ will be that last item found in our backtracking and complete
the path). Assume, without loss of generality, that $u_2 = u_{p1}$. The attacker
now knows that $L_2 = L_o(u_{p1})$ and vice versa. 
\subsection{NeNDS Shortcomings}
In addition to being susceptible to partial knowledge
reconstruction, the NeNDS algorithm has an exponential runtime which
is not suitable for real world applications. It can be shown the
NeNDS algorithm solves the Bottleneck Traveling Salesman problem
(BTSP), which is known to be NP-Complete in the general case. For
some cases of a defined distance function between values in the
database, such as in the case of one dimensional Euclidean distance,
there exists a polynomial time solution producing the same results as the NeNDS algorithm
(see Figure \ref{fig:nends-comparison} and Table
\ref{tbl:table-nends}).
\begin{figure}[hbt!]
\centering
\includegraphics[scale=0.5]{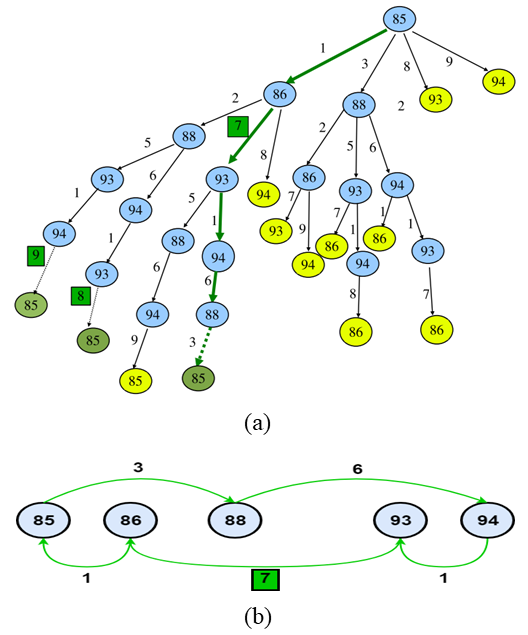}
\caption{Comparison of NeNDS original algorithm to bottleneck TSP for linear
networks. (a) NeNDS tree algorithm - exponential runtime. (b) Bottleneck TSP on
linear graph - linear runtime.}\label{fig:nends-comparison}
\end{figure}
\begin{table}[hbt!]
  \caption{NeNDS transformation result table - same for both algorithms.}
  \label{tbl:table-nends}  
  \begin{tabular}{|c|c|c|}
    \hline 
    Row & Original Value & Transformed Value \\
    \hline 
    $1$&$86$&$93$\\
    \hline 
    $2$&$88$&$85$\\
    \hline 
    $3$&$93$&$94$\\
    \hline 
    $4$&$85$&$86$\\
    \hline 
    $5$&$94$&$88$\\
  \hline 
\end{tabular}
\end{table}
In the general case, there are approximation algorithms for BTSP, such
as given by Kao and Sanghi \cite{btsp-approx} that can be adapted
for the case of NeNDS-like perturbation, giving the same level of
privacy while achieving an approximate level of accuracy.

\section{Conclusions}
With more and more user data being stored by companies and organizations, and a
growing demand to disseminate and share this data, the risks to security and
privacy rise greatly. While some of these issues have been addressed with
encryption and authorization protocols, the potential for misuse of access
still exists. The need for protecting user privacy while still maintaining
utility of databases has given birth to a wide variety of data anonymization
techniques. In this research we have analyzed the behavior, vulnerabilities and
shortcomings of instances of the data perturbation and the query-set-size
control methods. For query-set-size control over a vehicular network (graph
based dataset) we have shown the aggregate average query function to be
vulnerable to private data leakage. On the other hand, the amalgamation of the
minimum, maximum and median query functions allow for user privacy under the
examined model. For the case of data perturbation we have presented a
partial knowledge attack on the NeNDS algorithm. In addition, we prove the
exponential runtime of this algorithm and offer an alternative for the linear
coordinate case.\\
We plan on continued exploration of the privacy attacks described in
Section \ref{sec:combine}. We look to find an upper limit on the number of
allowed queries that will reduce the possibility of leaking private information.
Other directions in this area include analyzing this query behavior with regards
to different data structures, data types and query types. With regards to the
attack described in Section \ref{sec:cf-attack}, we intend to look for
possibilities of adding randomness to the base permutation of the original
data, in order to improve the privacy afforded by NeNDS. This may entail
selecting different random values for each $NH_{size}$ and using a
non-deterministic, non-optimal permutation algorithm for column transformation,
and could come at the cost of a minor loss of data accuracy.

\bibliographystyle{unsrt}
\bibliography{arxiv-privacy}
\end{document}